\documentclass[11pt]{article} 
\usepackage{fullpage}
\usepackage{times,epsfig,epsf,psfrag,array,amsmath,amssymb,verbatim}
\usepackage{url,comment,enumerate,graphicx,graphics,wrapfig}
\usepackage[usenames,dvipsnames,svgnames,table]{xcolor}
\usepackage[ruled,linesnumbered,vlined]{algorithm2e}
\usepackage{multirow,amsthm}
\usepackage{pdfpages,framed,tcolorbox}
\usepackage{tabularx}
\usepackage{pdflscape,latexsym,bm,array,caption,textcomp}

\usepackage{thmtools}
\usepackage{thm-restate}

\usepackage{hyperref}

\usepackage{cleveref}

\usepackage[square,sort,comma,numbers]{natbib}
\usepackage{url} 
\usepackage{hyperref}
\usepackage{tikz}
\usetikzlibrary{patterns}
\usetikzlibrary{decorations.pathreplacing}

\theoremstyle{plain}
\newtheorem{hypothesis}{Hypothesis}
\newtheorem{clm}{Claim}[section]

\theoremstyle{definition}
\newtheorem{definition}{Definition}[section]
\theoremstyle{definition}

\theoremstyle{definition}
\newtheorem{problem}{Problem}[section]

\theoremstyle{plain}
\newtheorem{theorem}{Theorem}[section]
\theoremstyle{plain}
\newtheorem{lemma}{Lemma}[section]
\theoremstyle{plain}

\theoremstyle{plain}
\newtheorem{proposition}{Proposition}[section]
\theoremstyle{plain}

\theoremstyle{plain}

\theoremstyle{remark}
\newtheorem{remark}{Remark}[section]

\theoremstyle{definition}

\AtBeginDocument{}

\newcommand{\W}{\mathcal{W}} 
\newcommand{\N}{\mathcal{N}}
 \newcommand{\calI}{\mathcal{I}}

\newcommand{\poly}{\text{poly}}

\newcommand{\NP}{\tt{NP}}

\newcommand{\dks}{D$k$S}

\newcommand{\mathbbm}[1]{\text{\usefont{U}{bbm}{m}{n}#1}}

\usepackage{csquotes}
\MakeOuterQuote{"}

\newcommand*{\email}[1]{%
	\normalsize\href{mailto:#1}{#1}\par
}

\begin{document}
\title{Hardness of Approximation of (Multi-)LCS over Small Alphabet}  

\author{Amey Bhangale\footnote{University of California Riverside, USA. Email: \email{ameyb@ucr.edu}}\and 
  Diptarka Chakraborty\footnote{National University of Singapore, Singapore. Author is supported in part by NUS ODPRT Grant, WBS No. R-252-000-A94-133. Email: \email{diptarka@comp.nus.edu.sg}} 
  \and Rajendra Kumar\footnote{IIT Kanpur, India and National University of Singapore. Author is supported in part
by the National Research Foundation Singapore under its AI Singapore
Programme [Award Number: AISG-RP-2018-005]. Email: \email{rjndr2503@gmail.com}}}

\pagenumbering{gobble}
\maketitle

\begin{abstract}
The problem of finding longest common subsequence (LCS) is one of the fundamental problems in computer science, which finds application in fields such as computational biology, text processing, information retrieval, data compression etc. It is well known that (decision version of) the problem of finding the length of a LCS of an arbitrary number of input sequences (which we refer to as Multi-LCS problem) is NP-complete. Jiang and Li [SICOMP'95] showed that if Max-Clique is hard to approximate within a factor of $s$ then Multi-LCS is also hard to approximate within a factor of $\Theta(s)$. By the NP-hardness of the problem of approximating Max-Clique by Zuckerman [ToC'07], for any constant $\delta>0$, the length of a LCS of arbitrary number of input sequences of length $n$ each, cannot be approximated within an $n^{1-\delta}$-factor in polynomial time unless {\tt{P}}$=${\NP}. However, the reduction of Jiang and Li assumes the alphabet size to be $\Omega(n)$. So far no hardness result is known for the problem of approximating Multi-LCS over sub-linear sized alphabet. On the other hand, it is easy to get $1/|\Sigma|$-factor approximation for strings of alphabet $\Sigma$.

In this paper, we make a significant progress towards proving hardness of approximation over small alphabet by showing a polynomial-time reduction from the well-studied \emph{densest $k$-subgraph} problem with {\em perfect completeness} to approximating Multi-LCS over alphabet of size $poly(n/k)$. As a consequence, from the known hardness result of densest $k$-subgraph problem (e.g. [Manurangsi, STOC'17]) we get that no polynomial-time algorithm can give an $n^{-o(1)}$-factor approximation of Multi-LCS over an alphabet of size $n^{o(1)}$, unless the Exponential Time Hypothesis is false.

\end{abstract}

\newpage
\pagenumbering{arabic}


 \section{Introduction}
 Finding \emph{longest common subsequence} (LCS) of a given set of strings over some alphabet is one of the fundamental problems of computer science. The computational problem of finding (the length of a) LCS has been intensively studied for the last five decades (see~\cite{Hirs83} and the references therein). This problem finds many applications in the fields of computational biology, data compression, pattern recognition, text processing and others. LCS is often considered among two strings, and in that case it is considered to be one of the classic string similarity measures (see~\cite{BHR00}). The general case, when the number of input strings is unrestricted, is also very interesting and well-studied. To avoid any confusion we refer to this general version of the LCS problem as \emph{Multi-LCS} problem. One of the major applications of Multi-LCS is to find similar regions of a set of DNA sequences. Multi-LCS is also a special case of the multiple sequence alignment and consensus subsequence discovery problem (e.g.~\cite{Pev92}). Interested readers may refer to the chapter entitled "Multi String Comparison-the Holy Grail" of the book~\cite{Gus97} for a comprehensive study on this topic. Other applications of Multi-LCS include text processing, syntactic pattern recognition~\cite{LF78} etc.
 
 Using a basic dynamic programming algorithm~\cite{WF74} we can find a LCS between two strings of length $n$ in quadratic time. However the general version, i.e., the Multi-LCS problem is known to be {\NP}-hard~\cite{Maier78} even for the binary alphabet. This problem remains {\NP}-hard even with certain restrictions on input strings (e.g.~\cite{BBJTV12}). For $m$ input strings a generalization of the basic dynamic programming algorithm finds LCS in time $O(m n^m)$. Recently, Abboud, Backurs and Williams~\cite{ABW15} showed that an $O(n^{m-\varepsilon})$ time (for any $\varepsilon >0$) algorithm for this problem would refute the Strong Exponential Time Hypothesis (SETH) even for alphabet of size $O(m)$.
 
 Due to the computational hardness of exact computation of a LCS, an interesting problem is what is the best approximation factor that we can achieve within a reasonable time bound. A $c$-approximate solution (for some $0<c \leq 1$) of a LCS is a common subsequence of length at least $c\cdot |LCS|$, where $|LCS|$ denotes the length of a LCS. For the Multi-LCS problem, Jiang and Li~\cite{jiang1995approximation} showed that if Max-Clique is hard to approximate within a factor of $s$ then Multi-LCS is also hard to approximate within a factor of $\Theta(s)$. By the NP-hardness of the problem of approximating Max-Clique by Zuckerman~\cite{Z07}, for any constant $\delta>0$, the length of a LCS of arbitrary number of input sequences of length $n$ each, cannot be approximated within an $n^{1-\delta}$-factor in polynomial time unless {\tt{P}}$=${\NP}. However, the result of  Jiang and Li~\cite{jiang1995approximation} is only true for alphabets of size $\Omega(n)$. For smaller alphabets (even for size sublinear in $n$) we do not know any such hardness result. Jiang and Li~\cite{jiang1995approximation} conjectured that Multi-LCS for even binary alphabet is {\tt{MAX-SNP}}-hard (see~\cite{PY91} for the definition of {\tt{MAX-SNP}}-hardness). To the best of our knowledge no progress has been done so far on the direction of showing any conditional hardness for smaller alphabets. On the other hand, it is very easy to get a $1/|\Sigma|$-approximation algorithm for the Multi-LCS problem over any alphabet $\Sigma$. The algorithm just outputs the best subsequence among the subsequences of the same symbol.
 
 In this paper, we make a significant progress towards showing hardness of approximation of Multi-LCS by refuting the existence of a polynomial time constant factor approximation algorithm under the Exponential Time Hypothesis (ETH). 
 
 
 \begin{theorem}
 \label{thm:main-ETH-hardness}
There exists a growing function $f(n) = n^{o(1)}$ such that assuming ETH, there is no polynomial time $\frac{1}{f(n)}$-factor approximation algorithm for the Multi-LCS problem over $n^{o(1)}$-sized alphabet.
 \end{theorem}
 
 This rules out {\em any} efficient poly-logarithmic factor approximation algorithm for the Multi-LCS problem over any $n^{o(1)}$-sized alphabet. We show the above theorem by providing a polynomial time reduction from the well-studied \emph{densest $k$-subgraph} problem with {\em perfect completeness} and its gap version $\gamma$-{\dks} (for the definition see Section~\ref{sec:prelims}).
 
 \begin{theorem}
 \label{thm:main-reduction}
Let $\frac{k}{n} = \frac{\beta(n)}{\gamma(n)}$ for $\beta<\gamma\leq 1$. If there is no polynomial time algorithm that solves $(\gamma^2/4)$-\dks($k,n$), then there is no polynomial time algorithm that solves $2\gamma$-approximate Multi-LCS problem over some alphabet of size $O(\frac{1}{\beta^6})$.
\end{theorem}
 The above reduction together with the ETH-based hardness result for the densest $k$-subgraph problem given by Manurangsi~\cite{manurangsi2017almost} implies Theorem~\ref{thm:main-ETH-hardness}. We refer to Appendix~\ref{sec:relatedworks} for the previous works related to the LCS problem and the densest $k$-subgraph problem.

 \subsection{Techniques}
 
Our reduction starts with the reduction from the Max-Clique problem to Multi-LCS given by~\cite{jiang1995approximation}. Given a graph $G$ on $n$ vertices the reduction outputs a Multi-LCS instance $\calI$ over an alphabet $\{a_1, a_2, \ldots, a_n\}$ of size $n$ with $2n$ strings. The reduction has a guarantee that the maximum LCS size of $\calI$ is {\em equal} to the size of the maximum clique in $G$.

A natural way to reduce the alphabet size is to replace each symbol $a_i$ in a string with a string $S_i \in \Sigma^m$ over a smaller alphabet $\Sigma$. Let us denote this new instance by $\calI'$.  The hope is that the only way to get a large LCS in $\calI'$ is to match the corresponding strings whenever the respective symbols in $\calI$ are matched. But this wishful thinking is not true when the alphabet size is much smaller than the original alphabet size as one might get a large common subsequence by matching parts of strings $S_i, S_j$ corresponding to the different symbols $a_i, a_j$ in the original strings.

We get away with this issue by using a special collection of strings $\{S_1, S_2, \ldots, S_n\}$ with the guarantee that for every pair $i\neq j$, LCS$(S_i, S_j)$ is much smaller than $m$. We can construct such a set deterministically by using the known deterministic construction of  the so called {\em long-distance synchronization strings}~\cite{HS18,CHLSW19}. There is also a much simpler randomized construction (see \Cref{thm:random-LCS}). It is easy to see that if the original strings have a LCS of size $t$, then the new Multi-LCS instance $\calI'$ over alphabet $\Sigma$ has an LCS of size at least $tm$. 

The interesting direction is to prove the converse i.e., if the LCS of $\calI'$ is large then the LCS of $\calI$ is also large. We do not know if this is true in general. So we rely on the starting problem of Max-Clique from which the instance $\calI$ (and hence $\calI'$) was created. We show that if $\calI'$ has large LCS, then we can find a large subgraph of $G$ which has a non trivial density (instead of finding a large clique). Thus, the reduction relies on hardness of approximation of the D$k$S problem with {\em perfect completeness}. Then we use the result of Manurangsi~\cite{manurangsi2017almost} which shows that given a graph $G$ with a guarantee that there is a clique of size $k$, there is no polynomial time algorithm which finds a subgraph of $G$ of size $k$ with density at least $\gamma(n)$ for some $\gamma(n) = o(n)$, assuming the ETH.

 \subsection{Related works}
 \label{sec:relatedworks}
 \subsubsection{Results on LCS problem}
 Finding LCS between two strings is an important problem in computer science. Wagner and Fischer~\cite{WF74} gave a quadratic time algorithm, which is in fact prototypical to dynamic programming. The running time was later improved to (slightly) sub-quadratic, more specifically $O(\frac{n^2\log \log n}{\log^2 n})$~\cite{MP80, G16}. Abboud, Backurs and Williams~\cite{ABW15} showed that a truly sub-quadratic algorithm ($O(n^{2-\varepsilon})$ for some $\varepsilon>0$) would imply a $2^{(1-\delta)n}$ time algorithm for CNF-satisfiability, contradicting the Strong Exponential Time Hypothesis (SETH). They in fact showed that for $m$ input strings an algorithm with running time $O(n^{m-\varepsilon})$ would refute SETH. Abboud et al.~\cite{AHWW16} later further strengthened the barrier result by showing that even shaving an arbitrarily large polylog factor from $n^2$ would have the plausible, but hard-to-prove, consequence that {\tt{NEXP}} does not have non-uniform ${NC}^1$ circuits. In case of approximation algorithm for LCS over arbitrarily large alphabets a simple sampling based technique achieves $O(n^{-x})$-approximation in $O(n^{2-2x})$ time. Very recently, an $O(n^{-0.497956})$ factor approximation (breaking $O(\sqrt{n})$ barrier) linear time algorithm is provided by Hajiaghayi et al.~\cite{HSSS19}. For binary alphabets another very recent result breaks $1/2$-approximation factor barrier in subquadratic time~\cite{RS20}. (Note, $1/|\Sigma|$-approximation over any alphabet $\Sigma$ is trivial.) The only hardness (or barrier) results for approximating LCS in subquadratic time are presented in~\cite{AB17, AR18}.
 
 For the general case (which we also refer as Multi-LCS), when the number of input strings is unrestricted, the decision version of the problem is known to be {\NP}-complete~\cite{Maier78} even for the binary alphabet. The problem remains {\NP}-complete even with further restriction like bounded \emph{run-length} on input strings~\cite{BBJTV12}. As cited earlier, Jiang and Li~\cite{jiang1995approximation} (along with the result of Zuckerman~\cite{Z07}) showed that for every constant $\delta>0$, there is no polynomial time algorithm that achieves $n^{1-\delta}$-approximation factor, unless {\tt{P}}$=${\NP}. One interesting aspect of the reduction in~\cite{jiang1995approximation} is that in any input string any particular symbol appears at most twice. It is worth mentioning that if we restrict ourselves to the input strings where a symbol appears exactly once, then we can find a LCS in polynomial time. The algorithm is just an extension of the dynamic programming algorithm that finds a longest increasing subsequence of an input sequence. It is also not difficult to show that the decision version of the Multi-LCS problem with the above restriction on the input strings can be solved even in non-deterministic logarithmic space. To see this, consider a LCS as a certificate. Then the verification algorithm makes single pass on the certificate, and checks whether every two consecutive symbols in the certificate appears in the same order in all the input strings. Clearly, the above verification algorithm uses only logarithmic space. Since we know that each symbol appears exactly once in a string, the above verification algorithm correctly decides whether the given certificate is a valid LCS or not.
 

 \subsubsection{Hardness results related to densest $k$-subgraph problem}
 Our starting point of the reduction is the hardness of approximating the densest $k$-subgraph problem. In the densest $k$-subgraph problem (D$k$S), we are given a graph $G(V,E)$ and an integer $1\leq k\leq |V|$. The task is to find a subgraph of $G$ of size $k$ with maximum density. Various approximation algorithms are known for D$k$S \cite{KP93, FPK01}, and the current best known is by \cite{BCCFV10} which gives $n^{1/4+\varepsilon}$-approximation algorithm for any constant $\varepsilon>0$.
 
 A special case of D$k$S is when it is guaranteed that $G$ has a clique of size $k$ and the task is to find a subgraph of size\footnote{Note, here \emph{size} of a subgraph refers to the number of vertices present in that subgraph.} $k$ with density at least $\gamma$ for $0<\gamma\leq 1$. In this {\em perfect completeness} case, Feige and Seltser \cite{FS97} gave an algorithm which finds a $k$ sized subgraph with density $(1-\varepsilon)$ in time $n^{O((1+\log \frac{n}{k})/\varepsilon)}$.

 There are several inapproximability results known for D$k$S based on worst-case assumptions. Khot~\cite{Khot06} ruled out a PTAS assuming $\NP\nsubseteq {\tt{BPTIME}}$ $(2^{n^\varepsilon})$ for some constant $\varepsilon>0$. Raghavendra and
 Steurer~\cite{raghavendra2010graph} showed that DkS is hard to approximate to within any constant ratio assuming the Unique Games Conjecture where the constraint graph satisfies a small set expansion property. 
 
 Assuming the Exponential Time Hypothesis, Braverman et al. \cite{BKRW17}, showed that for some constant $\varepsilon>0$, there is no polynomial time algorithm which when given a graph with a $k$-clique finds a $k$ sized subgraph with density $(1-\varepsilon)$. This result is significantly improved by Manurangsi~\cite{manurangsi2017almost} in which he showed that assuming ETH, no polynomial time algorithm can distinguish between the cases when $G$ has a clique of size $k$ and when every $k$ sized subgraph has density at most $n^{-1/(\log \log n)^c}$ for some constant $c>0$.


\section{Preliminaries}
\label{sec:prelims}
\paragraph*{Notations:} We use $[n]$ to denote the set $\{1,2,\cdots,n\}$. For any string $S$ we use $|S|$ to denote its length. By abuse of notation, for any set $V$ we also use the notation $|V|$ to denote the size of $V$. For any string $S$ of length $n$ and two indices $i,j \in [n]$, $S[i,j]$ denotes the substring of $S$ that starts at index $i$ and ends at index $j$. We use $\alpha(n), \beta(n), \gamma(n)$ to denote that $\alpha, \beta, \gamma$ are allowed to depend on $n$.

\subsection{Longest Common Subsequence}
Given $m$ sequences $S_1,\hdots,S_m$ of length $n$ over an alphabet $\Sigma$, the longest common subsequence is the longest sequence $S$ such that $\forall i\in [m], S$ is a subsequence of $S_i$.

We will refer to the computational problem of finding or deciding the length of LCS as a Multi-LCS problem. In this paper, we consider the decision variant of this problem: Given an integer $\ell\leq n$, we have to decide whether LCS has a length greater than equal to $\ell$, or less than $\ell$. For the approximation, we consider the following gap-version of this problem.

\begin{problem}
For any $0<\kappa<1$, the $\kappa$-approximate Multi-LCS problem is defined as: Given sequences $S_1,\hdots,S_m$ of length $n$ over an alphabet $\Sigma$ and an integer $\ell$, the goal is to distinguish between the following two cases
\begin{itemize}
	\item YES instance: A LCS of $S_1,\hdots,S_m$ has length greater than or equal to $\ell$.
	\item NO instance: A LCS of $S_1,\hdots,S_m$ has length less than $\kappa \cdot \ell$.
\end{itemize}
\end{problem}

We use the following definition of alignment.

\begin{definition}[Alignment]
Given two strings $S_1$ and $S_2$ of lengths $n$ and $m$ respectively, alignment $\sigma$ is a function from $[n]$ to $[m]\cup \{*\}$ which satisfies $\forall i\in[n],\text{if } \sigma(i)\neq * \text{ then } S_1[i]=S_2[\sigma(i)] $ and for any $i$ and $j$ if $\sigma(i)\neq *, \sigma(j)\neq *$ then for $i>j$, $\sigma(i)>\sigma(j)$.
\end{definition}
For an alignment $\sigma$ between two strings $S_1$ and $S_2$ we say $\sigma$ \emph{aligns} some subsequence $T_1=S_1[i_1]S_1[i_2]\cdots S_1[i_{\ell_1}]$ of $S_1$ with some subsequence $T_2=S_2[j_1]S_2[j_2]\cdots S_2[j_{\ell_2}]$ of $S_2$ if and only if for all $p \in [\ell_1]$, $\sigma(i_p) \in \{j_1,j_2,\cdots,j_{\ell_2}\}$.

\subsection{Exponential Time Hypothesis}
The Exponential Time Hypothesis (ETH) was introduced by Impagliazzo and Paturi~\cite{impagliazzo2001complexity}. It refutes the possibility of getting much faster algorithm to decide satisfiability of a $3$-CNF formula (also referred as $3$-SAT problem) than that by the trivial brute force method.

\begin{hypothesis}[ETH]
There is no $2^{o(n)}$ time algorithm for the $3$-SAT problem over $n$ variables.
\end{hypothesis}

%
%
%
%

\subsection{Densest $k$-Subgraph problem and related hardness results}
For any graph, the density is defined as the ratio of the number of edges present in it and the number of edges in any complete graph of the same size. So given a graph $G=(V,E)$, the density of $G$ is  $\frac{2|E|}{|V|^2-|V|} $.
	
The Densest $k$-Subgraph (\dks) problem is the following: Given a graph $G$ on $n$ vertices and a positive integer $k\leq n$, the goal is to find a subgraph of $G$ with $k$ vertices which has maximum density.

In this paper we will consider the following gap-version of densest $k$-subgraph, which in the literature is sometimes referred as densest $k$-subgraph with {\em perfect completeness}.

\begin{problem}
For any $\gamma \leq 1$, $\gamma$-\dks($k,n$) is defined as: Given a graph $G$ on $n$ vertices and a positive integer $k\leq n$, the goal is to distinguish between the following two cases
\begin{itemize}
	\item YES instance: There exists a clique of size $k$.
	\item NO instance: All subgraphs of size $k$ have density at most $\gamma$.
\end{itemize}
\end{problem}
We say that an algorithm solves $\gamma$-\dks($k,n$) if given any input it can distinguish whether the input is a YES instance or a NO instance. If the algorithm is randomized then it should succeed with probability at least $2/3$.

In this paper we use the following hardness result by Manurangsi~\cite{manurangsi2017almost}.
\begin{theorem}[\cite{manurangsi2017almost}]
\label{thm:ethclique}
There exists a constant $c_0 >0$ such that assuming the Exponential Time Hypothesis, for all constants $\varepsilon>0$, there is no polynomial time algorithm for $\gamma$-\dks($k,n$) where $\gamma = n^{-O\left(\frac{1}{(\log \log n)^{c_0}}\right)}$ and $\frac{k}{n} \in \left[n^{-\varepsilon}, n^{-\Omega\left(\frac{1}{\log \log n}\right)}\right]$. 
\end{theorem}


\section{Reduction}
\label{sec:reduction}
In this section we provide a reduction from the densest $k$-subgraph problem to the problem of approximating Multi-LCS and prove Theorem~\ref{thm:main-reduction}. Note that, Theorem~\ref{thm:main-reduction} and Theorem~\ref{thm:ethclique} together immediately imply Theorem~\ref{thm:main-ETH-hardness} by plugging $\gamma(n) =  n^{-O\left(\frac{1}{(\log \log n)^{c_0}}\right)}, \beta(n) = \gamma(n)^2$.

\begin{remark}
If we want to get the hardness of Multi-LCS for a constant sized alphabet using \Cref{thm:main-reduction}  then $k$ must be $\Omega(n)$. However, when $k=\Omega(n)$ \Cref{thm:ethclique} does not imply any hardness result. In fact, when $k=\Omega(n)$, there is a polynomial time algorithm for $(1-\varepsilon)$-D$k$S$(k,n)$ for any constant $\varepsilon>0$~\cite{FS97}. Therefore our reduction will not give any hardness for constant sized alphabet. However, if one can improve Theorem~\ref{thm:ethclique} for $k/n = 1/\poly(\log n)$ and $\gamma(n) = 1/\poly(\log n)$, then our main reduction in \Cref{thm:main-reduction} will imply Multi-LCS hardness for $\poly(\log n)$ sized alphabet!
\end{remark}

Our reduction involves two steps: First, we use the reduction from the Max-Clique problem to the Multi-LCS problem over large alphabet given in~\cite{jiang1995approximation}. Next we perform alphabet reduction by replacing each character by a "short" string over a small-sized alphabet.

\paragraph*{Revisiting the reduction from Max-Clique to Multi-LCS. }We first recall the reduction from~\cite{jiang1995approximation}. We are given a graph $G=(V,E)$ on $n$ vertices and an integer $k\leq n$. Fix an arbitrary labeling on the vertices of $V$ as $v_1,\hdots,v_n$. For every vertex $v_i$, partition its neighbors into two subsets: $\N_{<}(v_i)$ contains all the neighboring vertices $v_j$ with $j < i$; and $\N_{>}(v_i)$ contains all the neighboring vertices $v_j$ with $j > i$. 

Consider an alphabet $\Sigma$ containing a separate symbol for each vertex. We use $v_i$ to denote both the vertex and its corresponding symbol in $\Sigma$. Now for each vertex $v_i\in V$, construct the following two strings $X_i$ and $X'_i$   
\[ X_i=v_1 \hdots v_{i-1} v_{i+1} \hdots v_n v_i v_{i_r} \hdots v_{i_s}\text{ and } X'_i= v_{i_p} \hdots v_{i_q}v_i v_1 \hdots v_{i-1} v_{i+1} \hdots v_n  \]
where $\N_{>}(v_i)=\{v_{i_r}, \cdots, v_{i_s}\}$ with $i_r < \cdots < i_s$, and $\N_{<}(v_i)=\{v_{i_p}, \cdots, v_{i_q}\}$ with $i_p < \cdots < i_q$. The following proposition is immediate from the above construction.
\begin{proposition}[\cite{jiang1995approximation}]
\label{prop:complete}
If there is a clique of size $c$ in $G$, then there is a common subsequence of $X_1,\cdots,X_n$, $X'_1,\cdots,X'_n$ of length $c$.
\end{proposition}
The converse has also been shown in~\cite{jiang1995approximation}.
\begin{proposition}[\cite{jiang1995approximation}]
\label{prop:sound}
For any common subsequence $S$ of $X_1,\cdots,X_n,X'_1,\cdots,X'_n$, all the $v_i$'s present in $S$ form a clique in $G$.
\end{proposition}

The proofs of these propositions follow from the facts that any common subsequence is of the form $v_{i_1}, v_{i_2}, \ldots, v_{i_t}$ where $i_1<i_2<\ldots<i_t$ and that there must be an edge between $v_{i_j}$ and $v_{i_{j'}}$ for $1\leq j<j'\leq t$.

\paragraph*{Reducing the size of the alphabet. }For some parameter $\alpha(n)<1$, let $\{S_1,\hdots,S_n\}$ be a set of strings of length $m$ over some alphabet $\Sigma'$ such that: for all $i \ne j$
	$ |LCS(S_i,S_j)|\le \alpha m.$
We will fix the value of $m$ and $|\Sigma'|$ later. The following theorem (Theorem 1 of~\cite{KLM04}) shows that if we pick strings from $\Sigma'^{m}$ uniformly at random then for $|\Sigma'|=O(1/\alpha^2)$, with high probability the sampled strings will satisfy the above desired property.
\begin{theorem}[\cite{KLM04}]
\label{thm:random-LCS}
For every $\varepsilon >0$ there exists $c>0$ such that for large enough sized alphabet $\Sigma'$ for any $m$ if two strings $S_1,S_2$ are picked uniformly at random from $\Sigma'^m$ then 
$$Pr\Big[\Big||LCS(S_1,S_2)|-\frac{2m}{\sqrt{|\Sigma'|}}\Big| \ge \varepsilon\frac{2m}{\sqrt{|\Sigma'|}} \Big] \le e^{-cm/\sqrt{|\Sigma'|}}.$$
\end{theorem}
Now by suitably choosing $\varepsilon,m$ the following lemma directly follows from a union bound over every pair of  $n$ chosen strings.
\begin{lemma}
\label{lem:edit-codeword}
For any $\alpha\in (0,1)$, and $n\in \mathbb{N}$ there exists an alphabet $\Sigma'$ of size $O(\alpha^{-2})$ such that for any $m\ge c \alpha^{-1}\log n$ (for some suitably chosen constant $c>0$), if we choose a set of strings $S_1,\cdots,S_n$ uniformly at random from $\Sigma'^m$ then with probability at least $1-1/n$ for each $i\ne j$, $|LCS(S_i,S_j)| \le \alpha m$.
\end{lemma}
The above lemma gives us a randomized reduction. However we can deterministically find such a collection (with a slight loss in the parameters) using the known construction of \emph{synchronization strings}. The proof of the following Lemma is deferred to Appendix~\ref{sec:derand}.
\begin{restatable}{lemma}{missing}
\label{lem:edit-codeword-det}
For any $\alpha\in (0,1)$, and $n\in \mathbb{N}$ there exists an alphabet $\Sigma'$ of size $O(\alpha^{-3})$ such that for any $m > 2 \alpha^{-2}\log n$, there is a deterministic construction of a set of strings $S_1,\cdots,S_n \in \Sigma'^m$ such that for each $i\ne j$, $|LCS(S_i,S_j)| \le \alpha m$. Moreover, all the strings can be generated in time $O(\alpha^{-2}nm)$.
\end{restatable}

\begin{remark}
One advantage of using the randomized construction is the alphabet size (as well as the length of strings); randomized construction has only a quadratic loss whereas the deterministic construction has a cubic loss in the alphabet size. However this will not matter much for the parameters we need to prove our main theorem.
\end{remark}

Now let us continue with the description of our reduction. We replace each $v_j \in \Sigma$ by the string $S_j$. After the replacement we get the following two strings $Y_i$ and $Y'_i$ respectively from $X_i$ and $X'_i$.
\[ Y_i=S_1 \hdots S_{i-1} S_{i+1} \hdots S_n S_i S_{i_r} \hdots S_{i_s}\text{ and }
 Y'_i= S_{i_p} \hdots S_{i_q}S_i S_1 \hdots S_{i-1} S_{i+1} \hdots S_n   \]
Note, $Y_i$ and $Y'_i$'s are over the alphabet $\Sigma'$. For notational convenience we use $S_{\N_{>i}}$ to denote the substring $S_{i_r} \hdots S_{i_s}$, and $S_{\N_{<i}}$ to denote the substring $S_{i_p} \hdots S_{i_q}$. From now on, for simplicity, we will refer to these $S_i$'s as \emph{blocks}. Note, due to deterministic construction of strings $S_i$'s by Lemma~\ref{lem:edit-codeword-det} our whole reduction is deterministic and polynomial time. 

It follows directly from Proposition~\ref{prop:complete} that:
\begin{lemma}[Completeness]
\label{lem:completeness}
If graph $G$ is a YES instance of $\frac{\gamma^2}{4}$-$\dks$ (with clique of size $k$), then a LCS of $Y_1,\hdots,Y_n,Y'_1,\hdots,Y'_n$ is of length at least $km$.
\end{lemma}
We devote the rest of this section to proving the soundness of our reduction.
\begin{lemma}[Soundness]
\label{lem:soundness}
Let $\alpha \in (0, 1/8)$ and $\beta = \sqrt{8\alpha}$. If graph $G$ is a NO instance of $\frac{\gamma^2}{4}$-$\dks$ (every subgraph of size $k$ has density less than $\frac{\gamma^2}{4}$), then a LCS of $Y_1,\hdots,Y_n,Y'_1,\hdots,Y'_n$ has length at most $2\beta m n$.
\end{lemma}

\subsection{Proof of Soundness}
Let $L$ be an (arbitrary) LCS of $Y_1,\cdots,Y_n,Y'_1,\cdots,Y'_n$ of size greater than $2\beta mn$. By the construction $Y_n=S_1\hdots S_n$ (since $\N_{>}(v_n)=\emptyset$). So we can partition the subsequence $L$ as $Z_1,\cdots, Z_n$ where $\forall i\in [n]$ $Z_i$ is a subsequence of $S_i$. ($Z_i$ can be an empty string). Now consider all the $Z_i$ of length at least $\beta m$, and let $\W$ denote the set of all such $Z_i$'s, i.e., $\W=\{Z_i\mid |Z_i| \ge \beta m\}$. Suppose $L_1$ is the string formed by removing all $Z_i \not \in \W$ from $L$. Clearly, $|L_1|\geq |L|-\beta mn\geq \beta mn$.

For all $i,j \in [n]$ such that $i<j$, define $C[i,j]$ as:
$C[i,j]:= \{Z_t\in \W \mid  i \le t \le j \}.$
Note, $\W=C[1,n]$. Next we show that either the size of $C[1,n]$ is small or there exists a subgraph in $G$ which has large density.

Let us consider the set of vertices $V_H:=\{v_t|Z_t \in \W\}$. So $|V_H|=|\W|\ge \frac{|L|}{m}-\beta n \geq \beta n$. If we could show that the subgraph $H$ of $G$ induced by the set of vertices $V_H$ has high density (ideally, a clique), then that will imply Lemma~\ref{lem:soundness}. 

 Now consider an (arbitrary) alignment between $L_1$ and $Y_1,\cdots,Y_n,Y'_1,\cdots,Y'_n$. Let us denote the alignment between $L_1$ and $Y_i$ ($Y'_i$) by $\sigma_i$ ($\sigma'_i$). From now on whenever we will talk about alignment we will refer to these particular alignments ($\sigma_i$ or $\sigma'_i$ depending on strings under consideration) without specifying them explicitly. Consider a $Z_t\in \W$. We say $Z_t$ is \emph{$\varepsilon$-aligned} (for some $\varepsilon\in[0,1]$) with some substring $S'$ of some $Y_i$ (or $Y'_i$) if and only if either the first or the last $\varepsilon$ fraction of symbols of $Z_t$ is aligned by the alignment $\sigma'_i$ (or $\sigma'_i$) with some subsequence of $S'$. Throughout this proof we will set $\varepsilon=1/2$. Note that, if we partition $Y_i$ into (any) two parts $Y_i^l$ and $Y_i^r$ then $Z_i$ is $1/2$-\emph{aligned} to at least one of $Y_i^l$ and $Y_i^r$, and this justifies our setting of parameter $\varepsilon$.

By following the argument of the proof of Proposition~\ref{prop:sound} given in~\cite{jiang1995approximation}, it is possible to show that if $\sigma$ aligns all $Z_t$ with some subsequence of $S_t$ in all strings $Y_i$ (and $Y'_i$), then the subgraph $H$ induced by vertices in $V_H$ has high density (actually forms a clique). Unfortunately we do not know whether all the $Z_t$'s are aligned with their corresponding $S_t$'s in all the $Y_i$'s (and $Y'_i$'s). Following are the different cases of mapping $Z_i\in \W$ with $Y_i$: 
\begin{enumerate}
\item $Z_i$ is $1/2$\emph{-aligned} with the substring $S_1\hdots S_{i-1}$ of $Y_i$.
\item  $Z_i$ is $1/2$-aligned with $S_{i+1} \hdots S_n S_i S_{\N_{>i}} $ of $Y_i$ and there exists a $j>i$ such that a symbol of $Z_j$ in $L_1$ is aligned with some symbol of $S_j$ in the substring $S_{i+1}\hdots S_nS_i$. 
\item $Z_i$ is $1/2$-aligned with the substring $S_{i+1} \hdots S_n S_i S_{\N_{>i}} $ in $Y_i$ and there exists no $j>i$ such that a symbol of $Z_j\in \W$ is aligned with some symbol of $S_j$ in the substring $S_{i+1} \hdots S_n$.
\end{enumerate}
Similarly, we will also consider the mapping with $Y'_i$'s. We will categorize first and second case as \emph{sparse case} and the third one as the \emph{dense case}. Next we analyze these cases.

\subsubsection{Sparse Case: Improper mapping leads to small LCS locally}
 Let us recall that $Y_i=S_1 \hdots S_{i-1} S_{i+1} \hdots S_n S_i S_{\N_{>i}}
	\text{ and } 
 Y'_i=S_{\N_{<i}}S_i S_1 \hdots S_{i-1} S_{i+1} \hdots S_n.$ The next two claims demonstrate that if $Z_i$ is not mapped to $S_i$ in $Y_i$ (or $Y'_i$) then there is a portion $C[j,i]$ (or $C[i,j]$) in $L_1$ such that $\frac{|C[j,i]|}{i-j}$ (or $\frac{|C[i,j]|}{j-i}$) is small, i.e., that portion of $L_1$ is "sparse" with respect to the number of $Z_t$ blocks present in it.

\begin{clm}
\label{clm:sparse-case1}
If $Z_i\in \W$ is $1/2$-aligned with the substring $S_1\hdots S_{i-1}$ of $Y_i$ (by the alignment $\sigma_i$), then there exists a $j <i$ such that $|C[j,i]|\le \frac{2\alpha}{\beta} (i-j+1)$. Similarly, if $Z_i\in \W$ is $1/2$-aligned with the substring $S_{i+1}\hdots S_{n}$ of $Y'_i$ (by the alignment $\sigma'_i$), then there exists a $j >i$ such that $|C[i,j]|\le \frac{2\alpha}{\beta} (j-i+1)$. 
\end{clm}
\begin{proof}
Suppose $Z_i$ is $1/2$-aligned with $S_1\hdots S_{i-1}$ of $Y_i$. Let $j$ be the largest index less than $i$ such that a symbol in $Z_j$ is aligned (by $\sigma_i$) with some symbol in $S_j$ in $Y_i$ (if there does not exist such a $j$ then take $j=0$). Note, by the definition of $1/2$-alignment at least first $\beta m/2$ symbols of $Z_i$ are mapped (by $\sigma_i$) in $S_1 \hdots S_{i-1}$. Recall, the definition of $1/2$-alignment ensures the mapping of the first or the last half fraction of symbols. However in this case if $Z_i$'s last $\beta m/2$ symbols are mapped in $S_1\hdots S_{i-1}$ then the whole $Z_i$ is actually mapped in $S_1\hdots S_{i-1}$, which is even stronger than what we state. 

By the properties of strings $S_k$'s specified in Lemma~\ref{lem:edit-codeword-det}, the first $\beta m/2$ symbols of $Z_i$ require at least $\frac{\beta}{2\alpha}$ blocks from $\{S_{j}, S_{j+1}, \hdots, S_{i-1}\}$ to map completely (see Figure~\ref{fig:block-map}). 

\begin{figure}[h]
\begin{center}
\begin{tikzpicture}[yscale=0.4]

\draw (0,1) rectangle (10,2);

\draw (0,3) rectangle (10,4);

\draw (5,3) rectangle (6,4);
\node[] at (5.5,3.5) {$S_{i-1}$};

\draw (2,3) rectangle (3,4);
\node[] at (2.5,3.5) {$S_{t}$};

\draw (5,1) rectangle (6,2);
\node[] at (5.5,1.5) {$Z_i$};

\draw[dashed,color=red] (5,2)--(2,3);
\draw[dashed,color=red] (5.7,2)--(6,3);

\node[] at (-0.5,1.5) {$L_1$};

\node[] at (-0.5,3.5) {$Y_i$};

\draw [decorate,decoration={brace,amplitude=10pt},xshift=0pt,yshift=3pt, color=black,thick]
(2,4) -- (6,4) node [black,midway,yshift=0.5cm] {$\geq\frac{\beta}{2\alpha}$ blocks};

\end{tikzpicture}
\caption{$Z_i$ is $1/2$-aligned with $S_1\hdots S_{i-1}$ where $t>j$}
\label{fig:block-map}
\end{center}
\end{figure}
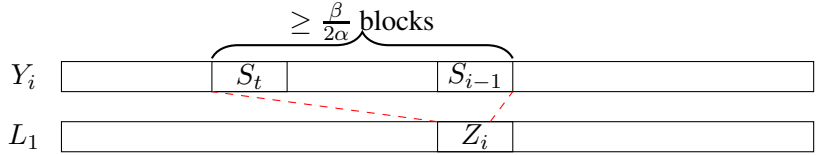

Similarly each element of $ C[j+1,i-1]$ also requires at least $\frac{\beta}{\alpha}$ blocks from $\{S_{j}, S_{j+1}, \hdots, S_{i-1}\}$. However any two $Z_p, Z_{p+1} \in C[j+1,i]$ may share a block (more specifically, the last block used for $Z_p$ and the first block used for $Z_{p+1}$) for mapping. So, we get
\begin{align*}
&\frac{\beta}{2\alpha}+(\frac{\beta}{\alpha}-1)|C[j+1,i-1]|\leq i-j \Rightarrow \frac{\beta}{2\alpha}|C[j+1,i]|\leq i-j.
\end{align*}
Note, $\frac{\beta}{\alpha}-1 \geq \frac{\beta}{2\alpha}$ as $\alpha\leq 1/8$ (recall, $\beta=\sqrt{8\alpha}$), and $C[j+1,i-1]\cup \{Z_i\}=C[j+1,i]$.


Similarly, suppose $Z_i$ is $1/2$-aligned with $S_{i+1}\hdots S_{n}$ of $Y'_i$. Let $j$ be the smallest index greater than $i$ such that a symbol of $Z_j$ is aligned (by $\sigma'_i$) with some symbol of $S_j$ in $Y'_i$ (if there does not exist any $j$ then take $j=n+1$). Using an argument similar to the above, we get
\begin{align*}
&\frac{\beta}{2\alpha}+(\frac{\beta}{\alpha}-1)|C[i+1,j-1]|\leq j-i \Rightarrow  \frac{\beta}{2\alpha}|C[i,j-1]|\leq j-i.
\end{align*}
\end{proof}


\begin{clm}
\label{clm:sparse-case2}
Suppose (by the alignment $\sigma_i$) $Z_i\in \W$ is $1/2$-aligned with $S_{i+1} \hdots S_n S_i S_{\N_{>i}} $ of $Y_i$, and there exists a $j>i$ such that a symbol of $Z_j$ in $L_1$ is aligned with some symbol of $S_j$ in the substring $S_{i+1}\hdots S_nS_i$. Then there exists $r$ such that $i < r \le j$ and $|C[i,r-1]|\leq \frac{2\alpha}{\beta}(r-i)$.

Similarly, suppose (by the alignment $\sigma'_i$) $Z_i\in \W$ is $1/2$-aligned with $S_{\N_{<i}}S_i S_1 \hdots S_{i-1} $ of $Y'_i$, and there exists a $j<i$ such that a symbol of $Z_j$ in $L_1$ is aligned with some symbol of $S_j$ in the substring $S_i S_1 \hdots S_{i-1}$. Then there exists $r$ such that $j \le r <i$ and $|C[r+1,i]|\leq \frac{2\alpha}{\beta}(i-r)$. 
\end{clm}
\begin{proof}
Suppose $Z_i$ is $1/2$-aligned with $S_{i+1} \hdots S_n S_i S_{\N_{>i}} $ of $Y_i$ and there exists a $j>i$ such that a symbol of $Z_j$ in $L_1$ is aligned (by $\sigma_i$) with some symbol of $S_j$ in the substring $S_{i+1}\hdots S_nS_i$. Let us choose $r$ to be the smallest $j$ with the above condition. By the argument used in the proof of Claim~\ref{clm:sparse-case1}, $Z_i$ requires at least $\frac{\beta}{2\alpha}$ blocks from $\{S_{i+1}, S_{i+2}, \cdots, S_{r}\}$, and every element in $C[i+1,r-1]$ requires at least $\frac{\beta}{\alpha}$ blocks from $\{S_{i+1}, S_{i+2}, \cdots, S_{r}\}$. Again, any two $Z_p, Z_{p+1} \in C[i,r-1]$ may share a block (more specifically, the last block used for $Z_p$ and the first block used for $Z_{p+1}$) for mapping. So we get
\begin{align*}
&\frac{\beta}{2\alpha}+ |C[i+1,r-1]|(\frac{\beta}{\alpha}-1)\leq r-i
\Rightarrow  \frac{\beta}{2\alpha}|C[i,r-1]|\leq r-i.
\end{align*}

Similarly, suppose $Z_i$ is $1/2$-aligned with $S_{\N_{<i}}S_i S_1 \hdots S_{i-1} $ of $Y'_i$ and there exists a $j<i$ such that a symbol of $Z_j$ in $L_1$ is aligned (by $\sigma'_i$) with some symbol of $S_j$ in the substring $S_i S_1 \hdots S_{i-1}$. Let us choose $r$ to be the largest $j$ with the above condition. Then we get 
\begin{align*}
&\frac{\beta}{2\alpha}+ |C[r+1,i-1]|(\frac{\beta}{\alpha}-1)\leq i-r
\Rightarrow  \frac{\beta}{2\alpha}|C[r+1,i]|\leq i-r.
\end{align*}
\end{proof}

\subsubsection{Dense Case: Proper mapping implies large number of neighbors}
Recall that $V_H=\{v_t \mid Z_t \in \W\}$. For each $v_i \in V_H$ further define $V_H^{>i}:=\{v_t\in V_H \mid t >i\}$ and $V_H^{<i}:=\{v_t\in V_H \mid t <i\}$. The next two claims show that if $Z_i$ is aligned with $S_i$ in $Y_i$ and $Y'_i$ then "most" of the vertices in $V_H$ are connected to (i.e., neighbors of) the vertex $v_i$. This eventually helps us to show that density of $H$ is high.
\begin{clm}
\label{clm:dense-forward}
 Suppose (by the alignment $\sigma_i$) $Z_i\in \W$ is $1/2$-aligned with $S_{i+1} \hdots S_n S_i S_{\N_{>i}} $ in $Y_i$, and there exists no $j>i$ such that a symbol of $Z_j\in \W$ is aligned with some symbol of $S_j$ in the substring $S_{i+1} \hdots S_n$. Then
$$|V_H^{>i} \bigcap \N_{>}(v_i)| + \frac{\beta}{2\alpha}|V_H^{>i} \setminus \N_{>}(v_i)|\le 2(n-i)+1.$$
\end{clm}
\begin{proof}
$Z_i$ is $1/2$-aligned with $S_{i+1} \hdots S_n S_i S_{\N_{>i}} $ of $Y_i$. So to align all $Z_r\in C[i+1,n]$ (note, $|C[i+1,n]|=|V_H^{>i}|$) at most $2(n-i)+1$ blocks of $S_p$'s are available. Since for no $j>i$ a symbol of $Z_j\in \W$ is aligned with some symbol of $S_j$ in $S_{i+1} \hdots S_n$, each $Z_r$ such that $v_r \in V_H^{>i} \setminus \N_{>}(v_i)$ requires at least $\frac{\beta}{\alpha}$ blocks of $S_p$'s to map. Any two $Z_r,Z_{r+1} $ such that $v_r,v_{r+1} \in V_H^{>i} \setminus \N_{>}(v_i)$ may share a block (more specifically, the last block used for $Z_p$ and the first block used for $Z_{p+1}$) for mapping. Recall for our choice of parameters $\alpha, \beta$, $\frac{\beta}{\alpha}-1 \geq \frac{\beta}{2\alpha}$. So we get 
$$|V_H^{>i} \bigcap \N_{>}(v_i)| + \frac{\beta}{2\alpha}|V_H^{>i} \setminus \N_{>}(v_i)|\le 2(n-i)+1.$$
\end{proof}
Similarly, we consider the mapping of $Z_i$ in the string $Y'_i$.
\begin{clm}
\label{clm:dense-back}
 Suppose (by the alignment $\sigma'_i$) $Z_i\in \W$ is $1/2$-aligned with $S_{\N_{<i}}S_i S_1 \hdots S_{i-1} $ in $Y'_i$, and there exists no $j < i$ such that a symbol of $Z_j\in \W$ is aligned with some symbol of $S_j$ in the substring $S_1 \hdots S_{i-1}$. Then
$$|V_H^{<i} \bigcap \N_{<}(v_i)| + \frac{\beta}{2\alpha}|V_H^{<i} \setminus \N_{<}(v_i)|\le 2i-1.$$
\end{clm}
\begin{proof}
$Z_i$ is $1/2$-aligned with $S_{\N_{<i}}S_i S_1 \hdots S_{i-1}$ of $Y'_i$. So to align all $Z_r \in C[1,i-1]$ (note, $|C[1,i-1]|=|V_H^{<i}|$), at most $2i-1$ blocks of $S_p$'s are available. Since for no $j<i$ a symbol of $Z_j\in \W$ is aligned with some symbol of $S_j$ in $S_1 \hdots S_{i-1}$, each $Z_r$ such that $v_r \in V_H^{<i} \setminus \N_{<}(v_i)$ requires at least $\frac{\beta}{\alpha}$ blocks of $S_p$'s to map. Any two $Z_r,Z_{r+1} $ such that $v_r,v_{r+1} \in V_H^{<i} \setminus \N_{<}(v_i)$ may share a block (more specifically, the last block used for $Z_p$ and the first block used for $Z_{p+1}$) for mapping. Recall for our choice of parameters $\alpha, \beta$, $\frac{\beta}{\alpha}-1 \geq \frac{\beta}{2\alpha}$. So we get 
$$|V_H^{<i} \bigcap \N_{<}(v_i)| + \frac{\beta}{2\alpha}|V_H^{<i} \setminus \N_{<}(v_i)|\le 2i-1.$$
\end{proof}

\subsubsection{Removing sparse blocks from LCS}
Next we choose a subset of vertices from the set $V_H$ so that the graph induced by that subset has high density. For that purpose we remove the "sparse" portions from the subsequence $L_1$ in the following way:
\begin{enumerate}
\item Initialize an empty set $\mathcal{T}$.
\item For each $Z_i \in \W$ identify the largest $j>i$ such that $\frac{|C[i,j]|}{j-i+1}\le \frac{2\alpha}{\beta}$, and then add all $Z_k \in C[i,j]$ in the set $\mathcal{T}$. (If no such $j$ exists then do not add anything to $\mathcal{T}$.)
\item Define a new set $\W'=\W\setminus \mathcal{T}$.
\end{enumerate}
Let $L_2$ be the string formed by removing all $Z_i \not \in \W'$ from $L_1$. Let us also define a set of vertices $V'_H=\{v_t | Z_t \in \W'\}$. (Note, $V'_H\subseteq V_H$.) Now we will argue that the set $V_H$ has not shrunk by much after removing the sparse blocks and each vertex in $V'_H$ has high degree in the subgraph $H$, which eventually implies that the subgraph $H$ has high density.
\begin{clm}
\label{clm:large-subgraph}
$|V'_H| \ge |V_H| - \frac{4\alpha}{\beta}n$.
\end{clm}
\begin{proof}
	Let us consider the set $\mathcal{T}$. We can write $\mathcal{T}$ as a union of disjoint subsets as $\mathcal{T}=C[p_1,q_1]\cup C[p_2,q_2]\cup \cdots \cup C[p_{r},q_{r}]$ for some integer $r \in [n]$, such that $\forall_{1\le \ell \le r-1}$ $C[q_{\ell},p_{\ell+1}]\ne \emptyset$ (see Figure~\ref{fig:split-T}).

	\begin{figure}
		\begin{center}
			
			\begin{tikzpicture}[xscale=0.8, yscale=0.7 ]
			
			\draw (0,0) rectangle (10,1);
			\draw[pattern=north east lines, pattern color=red] (2,0) rectangle (5,1);
			\draw[pattern=north east lines, pattern color=red] (6,0) rectangle (8,1);
			\node[] at (1.8,0.8) {$p_1$};
			\node[] at (5.2,0.8) {$q_1$};
			\node[] at (5.8,0.8) {$p_2$};
			\node[] at (8.2,0.8) {$q_2$};

			\draw (1,-1) rectangle (9,-2); 
			\draw[dashed] (1,-1) -- (0,-1);
			\draw[dashed] (1,-2) -- (0,-2);
			\draw[dashed] (9,-1) -- (10,-1);
			\draw[dashed] (9,-2) -- (10,-2);
			\draw[color=green] (3,-1) -- (3,-2);
			\draw[color=red] (4,-1) -- (4,-2);
			\draw[color=blue] (6,-1) -- (6,-2);
			\draw[color=green] (7,-1) -- (7,-2);
			\draw[dashed] (2,0) -- (1,-1);
			\draw[dashed] (5,0) -- (9,-1);
			
			\draw[color=red] (1,-1) -- (1,-2);
			\node[] at (0.8,-1.8) {$p_1$};
			\node[] at (0.8,-2.2) {$i_1$};
			
			\node[] at (2.8, -1.8) {$i_2$};
			
			\node[] at (4.2,-2.2) {$j_1$};
			
			\node[] at (5.8,-2.2) {$i_3$};
			
			\node[] at (7.2,-1.8) {$j_2 $};
			
			\node[] at (9.2,-1.8) {$q_1$};
			
			\node[] at (9.2,-2.2) {$j_3$};

			\draw [decorate,decoration={brace,amplitude=10pt},xshift=0pt,yshift=-1pt, color=red,thick]
			(4,-2) -- (1,-2) node [black,midway,yshift=--1cm] {};
			
			\draw [decorate,decoration={brace,amplitude=10pt},xshift=0pt,yshift=2pt, color=green,thick]
			(3,-1) -- (7,-1) node [black,midway,yshift=-0.5cm] {};
			
			\draw[color=blue] (9,-1) -- (9,-2);
			\draw [decorate,decoration={brace,amplitude=10pt},xshift=0pt,yshift=-1pt, color=blue,thick]
			(9,-2) -- (6,-2) node [black,midway,yshift=--1cm] {};
			
			\draw [decorate,decoration={brace,amplitude=10pt},xshift=0pt,yshift=3pt, color=black,thick]
			(0,1) -- (10,1) node [black,midway,yshift=-0.5cm] {};
			
			\node[] at (5,1.8) {${C[1,n]}$};

			\node[text width=5cm] at (14,1) {Shaded region is included in $\mathcal{T}$};
			
			\node[text width=5cm] at (14,-1.1) {Considering $s=3$, $(i_1,j_1)$,$(i_2,j_2)$,$(i_3,j_3)$ is a series of pairs to cover $C[p_1,q_1]$ where $i_1=p_1$ and $j_3=q_1$ };

			\end{tikzpicture}
			
			\caption{$\mathcal{T}$ as a union of disjoint subsets}
			\label{fig:split-T}
		\end{center}
		
	\end{figure}
	
	Now if we could show that for each $\ell \in [r]$, $|C[p_{\ell},q_{\ell}]| \le \frac{4\alpha}{\beta}(q_{\ell}-p_{\ell})$, then
	\begin{align*}
	|\mathcal{T}| & = \sum_{\ell=1}^r |C[p_{\ell},q_{\ell}]| \le \frac{4 \alpha}{\beta} \sum_{\ell=1}^r (q_{\ell}-p_{\ell}) \le \frac{4 \alpha}{\beta}n
	\end{align*}
	where the last inequality is true since $p_1<q_1<p_2<q_2<\cdots<p_{r}<q_r$. So to conclude the proof of the claim next we show that for all $\ell \in [r]$ $|C[p_{\ell},q_{\ell}]| \le \frac{4\alpha}{\beta}(q_{\ell}-p_{\ell})$.
	
	It is immediate from the construction of the set $\mathcal{T}$ that there exists a sequence of pair of indices $(i_1,j_1),\cdots,(i_s,j_s)$ (for some positive integer $s$) where $i_1=p_{\ell}$ and $j_s=q_{\ell}$, such that for all $t \in [s]$ while processing $Z_{i_t}$ we add blocks of $C[i_t,j_t]$ in $\mathcal{T}$, and $C[p_{\ell},q_{\ell}] = \bigcup_{t \in [s]}C[i_t,j_t]$. We can further assume that there exists no $t'\in [s]$ such that $C[i_{t'},j_{t'}] \subseteq \bigcup_{t \in [s]\setminus\{t'\}}C[i_t,j_t]$. (In words it means that $C[i_1,j_1],\cdots,C[i_s,j_s]$ is a minimal sequence of subsets whose union is $C[i_1,j_s]$.) Due to this assumption we can write that $i_2 \le j_1 \le i_3 \le j_2 \le \cdots \le i_s \le j_{s-1}$ and $\forall t\in [s-2], i_{t+2}\geq j_t +1$ (see Figure~\ref{fig:split-T}).
	So,
	\vspace{-9pt}
	\begin{align*}
	|C[p_{\ell},q_{\ell}]| \le \sum_{t=1}^{s} |C[i_t,j_t]| &\le \frac{2\alpha}{\beta} \sum_{t=1}^{s} (j_t-i_t+1)\\
	 & = \frac{2 \alpha}{\beta} \Big[s+ (j_s-i_1) + \sum_{t=1}^{s-1}(j_t-i_{t+1})\Big]\\
	& \leq \frac{2 \alpha}{\beta} \Big[s+ (j_s-i_1) + (j_{s-1}-i_{2} -(s-2) )\Big]\\
	& \leq \frac{2 \alpha}{\beta} \Big[ 2(j_s-i_1) \Big]
	\end{align*}
	where second last inequality uses the fact that $\forall t\in [s-2], i_{t+2}\geq j_t +1 $ and last inequality uses the fact that $j_s\geq j_{s-1}+1$ and $i_2\geq i_1+1$. 
	Hence we conclude that $|C[p_{\ell},q_{\ell}]| \le \frac{4\alpha}{\beta}(q_{\ell}-p_{\ell})$, and this completes the proof.
\end{proof}

\begin{clm}
\label{clm:high-density}
For each vertex $v_i \in V'_H$, $|V_H \bigcap \N(v_i)| \ge |V_H| - \frac{4 \alpha}{\beta}n$.
\end{clm}
\begin{proof}
By the construction of $\W'$, for each $Z_i \in \W'$ we know that there exists no $j >i$ (or $<i$) such that $\frac{|C[i,j]|}{j-i+1} \le \frac{2\alpha}{\beta}$ (or $\frac{|C[j,i]|}{i-j+1} \le \frac{2\alpha}{\beta}$). Then by Claim~\ref{clm:sparse-case1} and Claim~\ref{clm:sparse-case2} it follows that all $Z_i \in \W'$ satisfy preconditions of both Claim~\ref{clm:dense-forward} and Claim~\ref{clm:dense-back}. Otherwise by Claim~\ref{clm:sparse-case1} and Claim~\ref{clm:sparse-case2} we know that there exists a $j >i$ (or $<i$) such that $\frac{|C[i,j]|}{j-i+1} \le \frac{2\alpha}{\beta}$ (or $\frac{|C[j,i]|}{i-j+1} \le \frac{2\alpha}{\beta}$). For $j>i$ when we process $Z_i$ to construct the set $\mathcal{T}$ we add all the blocks of $C[i,j]$, and for $j<i$ when we process $Z_j$ we add all the blocks of $C[j,i]$. So it must be the case that the alignment $\sigma_i$ between $L_1$ and $Y_i$, $1/2$-aligns $Z_i$ to the substring $S_{i+1} \hdots S_n S_i S_{\N_{>i}} $ and there exists no $j>i$ such that $Z_j\in \W$ aligns with $S_j$ in the substring $S_{i+1} \hdots S_n$. Also, $\sigma'_i$ $1/2$-aligns $Z_i$ to the substring $S_{\N_{<i}}S_i S_1 \hdots S_{i-1} $ and there exists no $j < i$ such that $Z_j\in \W$ aligns with $S_j$ in the substring $S_1 \hdots S_{i-1}$. So by Claim~\ref{clm:dense-forward}
$$|V_H^{>i} \bigcap \N_{>}(v_i)| + \frac{\beta}{2\alpha}|V_H^{>i} \setminus \N_{>}(v_i)|\le 2(n-i)+1,$$
and by Claim~\ref{clm:dense-back}
$$|V_H^{<i} \bigcap \N_{<}(v_i)| + \frac{\beta}{2\alpha}|V_H^{<i} \setminus \N_{<}(v_i)|\le 2i-1.$$
These two claims together imply
\begin{align*}
&|V_H \bigcap \N(v_i)| +\frac{\beta}{2\alpha}|V_H \setminus \N(v_i)| \le 2n\\
\Rightarrow & |V_H \bigcap \N(v_i)| +\frac{\beta}{2\alpha}(|V_H| - |V_H \bigcap \N(v_i)|) \le 2n\\
\Rightarrow & (\frac{\beta}{2 \alpha}-1)|V_H \bigcap \N(v_i)| \ge \frac{\beta}{2 \alpha} |V_H|-2n\\
\Rightarrow & |V_H \bigcap \N(v_i)| \ge |V_H|-\frac{4 \alpha}{\beta}n.
\end{align*}
\end{proof}

Now we are ready to complete the proof of soundness (Lemma~\ref{lem:soundness}).

\begin{proof}[Proof of Lemma~\ref{lem:soundness}]
For the sake of contradiction let us assume that the LCS is of size at least $2\beta mn$. Recall, we have already seen that $|V_H| \ge \beta n$. Now we consider the following two cases depending on the size of $V_H$.

\paragraph*{Case 1: (When $|V_H|\leq \frac{\beta}{\gamma}n$)} 
Suppose $|V_H|\leq \frac{\beta}{\gamma}n$ ($=k$). Let $V'\supseteq V_H$ be an arbitrary set of size exactly $\frac{\beta}{\gamma} n$. Let $H'$ be the subgraph induced by the vertices $V'$. Using Claim~\ref{clm:large-subgraph} and Claim~\ref{clm:high-density}, we can lower bound the density of the subgraph $H'$ by:
\begin{align*}
     \frac{\frac{1}{2}\sum_{v\in V'_H} \left(|V_H| - \frac{4 \alpha}{\beta}n\right)}{ {|V'|\choose 2}}
    \geq \frac{\left( \beta - \frac{4\alpha}{\beta}\right)n \cdot \left( \beta - \frac{4\alpha}{\beta}\right)n}{ \frac{\beta}{\gamma} n \cdot \frac{\beta}{\gamma} n}
     \geq  \left( \gamma - \frac{4\alpha\gamma}{\beta^2}\right)^2.
\end{align*}
As we set $\alpha = \beta^2/8$, we get that the density of the subgraph induced by $V'$ is at least $(\gamma/2)^2$.

\paragraph*{Case 2: (When $|V_H|>  \frac{\beta}{\gamma}n$)} 
If $|V_H|>  \frac{\beta}{\gamma}n$, the density of the subgaph $H$ induced by $V_H$ is lower bounded by:
\begin{align*}
\frac{\frac{1}{2}\sum_{v\in V'_H} \left(|V_H| - \frac{4 \alpha}{\beta}n\right)}{ {|V_H|\choose 2}} &\ge \frac{|V'_H|\left(|V_H| - \frac{4 \alpha}{\beta}n\right)}{|V_H|(|V_H|-1)}\\
& \ge \frac{\left(|V_H| - \frac{4 \alpha}{\beta}n\right)^2}{|V_H|^2}\\
& = \left(1-\frac{4 \alpha n}{\beta |V_H|}\right)^2\\
& \ge (1-\frac{\gamma}{2})^2 \qquad\qquad\text{(since $|V_H|>  \frac{\beta}{\gamma}n$ and we set $\alpha = \beta^2/8$)}\\
& \ge (\gamma/2)^2\qquad \qquad\quad \text{(since $\gamma \le 1$)}.
\end{align*}
Now since density of the subgraph is at least $(\gamma/2)^2$, it follows from the following simple claim that there exists a subgraph of $H$ of size $\frac{\beta}{\gamma} n$ which has density at least $(\gamma/2)^2$. 

\begin{clm}
	Suppose a graph $G=(V, E)$ has edge density $c$, then for any $2\leq k \leq |V|$, there exists a subgraph of size $k$ with density at least $c$.
	\end{clm}
\begin{proof}
	Let $n = |V|$. Pick a subset $H\subseteq V$ of size exactly $k$ uniformly at random. For a fixed edge $e$ in $G$, the probability that the edge $e$ is present in the subgraph induced by $H$ is exactly $\frac{{n-2 \choose k-2}}{{n\choose k}}$. Since $G$ has $c\cdot {n \choose 2}$ edges, by linearity of expectation, the expected number of edges in the subgraph induced by $H$ is equal to $c\cdot{n \choose 2}\cdot \frac{{n-2 \choose k-2}}{{n\choose k}} = c\cdot {k \choose 2}$. Therefore, the expected density of the subgraph is exactly equal to $c$. Hence, by an averaging argument, there exists a subgraph of $G$ of size $k$ with density at least $c$.
\end{proof}

In both the cases, we have shown that there exists a subgraph of size $\frac{\beta}{\gamma} n(=k)$ with density at least $(\gamma/2)^2$, which is a contradiction to the fact that we started with a NO instance of $\frac{\gamma^2}{4}$-\dks$\left(\frac{\beta}{\gamma} n, n\right)$. Therefore in this case, the size of LCS must be at most $2\beta mn$.
\end{proof}

\paragraph{Proof of Theorem~\ref{thm:main-reduction}:}
If there is no polynomial time algorithm to distinguish between the YES and NO instances of $\frac{\gamma^2}{4}$-\dks$\left(\frac{\beta}{\gamma} n, n\right)$, then using Lemma~\ref{lem:completeness} and Lemma~\ref{lem:soundness}, it follows that there is no polynomial time algorithm to distinguish between the cases when the LCS of $Y_1,\cdots,Y_n,Y'_1,\cdots,Y'_n$ is of size $\frac{\beta}{\gamma}mn$ vs. $2\beta mn$. Also note that if we use Lemma~\ref{lem:edit-codeword-det} to construct the strings $S_i$'s then the alphabet size is $O(\alpha^{-3}) = O(\beta^{-6})$. This proves the main theorem.


\section{Conclusion}
In this paper we show hardness of constant factor approximation of Multi-LCS problem with input of length $n$ over $n^{o(1)}$ sized alphabet assuming the Exponential Time Hypothesis (ETH). This is the first hardness result for approximating Multi-LCS problem for sublinear sized alphabet. To prove our result we provide a reduction from the densest $k$-subgraph problem with perfect completeness, and then use the known hardness results for the latter problem from~\cite{manurangsi2017almost}. One interesting fact is that if one could show hardness of the $\gamma$-\dks$(k, n)$ problem for $k=\Theta(\frac{n}{poly \log n})$ and $\gamma = (\log n)^{-c}$ for some $c>0$,  then due to our reduction that will directly imply constant factor hardness for Multi-LCS over poly-logarithmic sized alphabet under ETH.

\paragraph*{Acknowledgements. }Authors would like to thank anonymous reviewers for providing helpful comments on an earlier version of this paper and especially for pointing out a small technical mistake in the proof of Lemma~\ref{lem:soundness}. Authors would also like to thank Pasin Manurangsi for pointing out that for certain regimes no hardness result is known for the densest $k$-subgraph problem.

\bibliographystyle{plain}
\bibliography{ref}

\appendix

\section{Derandomized version of Lemma~\ref{lem:edit-codeword}}
\label{sec:derand}
To achieve deterministic reduction we need to construct the set of strings $S_1,\cdots,S_n $ deterministically in time $poly(n)$. For that purpose we use the notion of \emph{synchronization strings} used in the literature of \emph{insertion-deletion codes}~\cite{HS18,CHLSW19}.

\begin{definition}[$c$-long-distance $\varepsilon$-synchronization string]
\label{def:sync-string}
A string $S\in \Sigma^n$ is called a $c$-long-distance $\varepsilon$-synchronization string for some parameter $\varepsilon \in (0,1)$, if for every $1\le i<j\le i'<j'\le n$ with $i'-j \le n\cdot \mathbbm{1}_{(j+j'-i-i')>c\log n}$, $|LCS(S[i,j],S[i',j'])| \le \varepsilon (j+j'-i-i')$, where $\mathbbm{1}_{(j+j'-i-i')>c\log n}$ is the indicator function for $(j+j'-i-i')>c\log n$.
\end{definition}
Note, in the definition of $c$-long-distance $\varepsilon$-synchronization string in~\cite{CHLSW19} authors used the notion of \emph{edit distance} instead of LCS. More specifically, they specified the edit distance between $S[i,j]$ and $S[i',j'])$ is at least $(1-\varepsilon) (|S[i,j]|+|S[i',j']|)$. However both the notions can be used interchangeably since for any two strings $S,S'$, $|LCS(S,S')|=|S|+|S'|-ED(S,S')$, where the edit distance $ED(S,S')$ is defined as the minimum number of insertion and deletion operations required to transform $S$ to $S'$. One may note that, generally while defining the edit distance we also allow substitution operation. However here we are not allowing substitution operation, and that is why we are able to write the following equivalence between LCS and the edit distance of two strings $S,S'$: $|LCS(S,S')|=|S|+|S'|-ED(S,S')$. We would like to mention that in~\cite{CHLSW19} authors also used this particular version of the edit distance notion (i.e., without substitution operation).

Several constructions of such long-distance synchronization strings are given in~\cite{HS18,CHLSW19} with different parameters. However we restate one of the theorems from~\cite{CHLSW19} that we find useful for our purpose.

\begin{theorem}[Rephrasing of Theorem 5.4 of~\cite{CHLSW19}]
\label{thm:sync-const}
For any $n \in \mathbb{N}$ and parameter $\varepsilon \in (0,1)$, there is a deterministic construction of an $\varepsilon^{-2}$-long-distance $\varepsilon$-synchronization string $S \in \Sigma^n$ for some alphabet $\Sigma$ of size $O(\varepsilon^{-3})$. Moreover, for any $i \in [n]$ the substring $S[i,i+\log n]$ can be computed in time $O(\varepsilon^{-2}\log n)$.
\end{theorem}

Now using the above we will provide deterministic construction of set of strings $S_1,\cdots,S_n$ with our desired property.
\missing*
\begin{proof}
For a specified $\alpha$ and $n$, set $\varepsilon=\alpha/2$. Then use the construction from Theorem~\ref{thm:sync-const} to get an $\varepsilon^{-2}$-long-distance $\varepsilon$-synchronization string $S$ of length $2nm$, for any $m > \frac{1}{2}\varepsilon^{-2} \log n$. The bound on $m$ is required to satisfy the condition that $(j+j'-i-i')>c\log n$ of Definition~\ref{def:sync-string}. (Note, in our case $(j+j'-i-i')=2m$ and $c=\varepsilon^{-2}$.) Then divide the string $S$ into $m$ length blocks. Finally choose alternate blocks as $S_1,\cdots,S_n$. More specifically, $S_1=S[1,m], S_2=S[2m+1,3m],\cdots, S_n=S[(2n-2)m+1,(2n-1)m]$. Now the bound on $|LCS(S_i,S_j)|$ for any $i \ne j$, directly follows from Definition~\ref{def:sync-string}.
\end{proof}

\end{document}